\documentclass[10pt,letterpaper]{article}
\pdfoutput=1
\usepackage{ifpdf}
\usepackage[utf8]{inputenc}
\usepackage{amsmath}
\usepackage{amsfonts}
\usepackage{amssymb}
\usepackage{float}
\usepackage{graphicx}
\usepackage{cancel}
\usepackage[margin=1in]{geometry}
\usepackage{afterpage}
\usepackage{xcolor}
\usepackage{subfig}
\usepackage{tabularx}
\newcolumntype{Y}{>{\centering\arraybackslash}X}
\usepackage{multirow}
\usepackage{hhline,colortbl}
\usepackage{fancyhdr}
\usepackage{amsthm}

\providecommand{\U}[1]{\protect\rule{.1in}{.1in}}
\newtheorem{lemma}{Lemma} 
\newtheorem{prop}{Proposition}
\newtheorem{cor}{Corollary} 
\usepackage[normalem]{ulem}

\author{Xiaotian Xu\thanks{Authors are with the Department of Mechanical Engineering,
		2181 Glenn L. Martin Hall, Building 088,
		University of Maryland,
		College Park, MD 20742, USA.
		Email: xxu0116@umd.edu, yancy@umd.edu.} 
	~and~Yancy Diaz-Mercado
}%
\title{\LARGE \bf
	Multi-Agent Control Using Coverage Over Time-Varying Domains}
\date{\vspace{-5ex}}

\begin{document}
\maketitle

\begin{abstract} 
	
	Multi-agent coverage control is used as a mechanism to influence the behavior of a group of robots by introducing time-varying domain. The coverage optimization problem is modified to adopt time-varying domains, and the proposed control law possesses an exponential convergence characteristic. Cumbrous control for many robots is simplified by deploying distribution and behavior of the robot team as a whole.
	In the proposed approach, the inputs to the multi-agent system, i.e., time-varying density and time-varying domain, are agnostic to the size of the system. Analytic expressions of surface and line integrals present in the control law are obtained under uniform density. The scalability of the proposed control strategy is explained and verified via numerical simulation. Experiments on real robots are used to test the proposed control law.  
\end{abstract}%


\section{Introduction}
Coverage control of multi-agent system has been drawing attention for a long time due to its wide applications, such as surveillance or exploration of an interested region. However, most of works consider these regions are static. We investigate the case that the interested areas are dynamic in this paper and extend the potential applications of coverage control of multi-agent systems. Typically, coverage control employs a group of robots to optimally cover an interested area. A classical solution to coverage problems is proposed by involving proper partitions of the domain\cite{lloyd1982least,cortes2005spatially}. Coverage control based schemes have been previously developed for multi-agent systems, e.g., broadcast control (BC) scheme
\cite{darmaraju2018coverage} and pivot-based collective coverage algorithm
\cite{luo2018pivot}. In \cite{leeTVDC, SwarmBook} the authors introduce a  mechanism to influence the collaborative behavior between robots in coverage control by associating time-varying densities to the domain. An approach to perform coverage control on moving regions is discussed in \cite{nishigaki2019coverage}. However, the robot team is not able to maintain the centroidal Voronoi tessellation (CVT) configuration of the moving region, which will be one of the topics considered in this paper.

In this paper, the coverage control algorithm \cite{leeTVDC} is taken as a mechanism to affect the behavior of a big multi-robot team by involving time-varying domains. We define a time-varying subdomain in the workspace, and command the robots to optimally distribute in the time-varying subdomain. Under this strategy, when the subdomain moves and the shape or the scale of the subdomain changes, the robot team will act in a coordinating manner. The proposed control strategy allows the multi-objective control of multi-agent systems to be simplified into manipulating the subdomain directly, e.g., patrolling along the trajectory, and changing the size of robot group. Nevertheless, the formation of the robot team can be controlled by adopting time-varying densities within the subdomain as suggested in \cite{ SwarmBook,Timevaryingdomains}.

The outline of this study is as follows. In Section \ref{Preliminaries}, we
state the coverage problem and recall pertinent prior work. The
coverage control over time-varying domains strategy is proposed in Section
\ref{TimeVaryingDomains}. Analytic solutions are presented in Section
\ref{AnalyticalFormulas} under the uniform density. The scalability of the proposed algorithm is presented and validated in Section \ref{Scalability}. A robotic implementation
of the proposed control scheme together with results of experiments are illustrated in
Section \ref{Implementation}. Finally, conclusions are summarized
in Section \ref{Conclusion}.


\section{Preliminaries}

\label{Preliminaries} Coverage control will be used as a mechanism to influence
the behavior of the multi-robot team, to coordinate their inter-robot motion,
and manipulate their motion in the workspace. In this section, we provide some
preliminary descriptions of the coverage control problem that is addressed with
the modifications needed to account for time-varying domains.\vspace{-0.5em}

\subsection{The Coverage Problem}

Let $p_{i}\in\mathcal{D}\subseteq\mathbb{R}^{d}$ be the position of the
$i$\textsuperscript{th} robot, $i\in\{1,\ldots,n\}$, in the domain of
interest $\mathcal{D}$, i.e., the robot workspace. Further define a convex
time-varying \emph{subdomain} $\mathcal{S}(t)\subset\mathcal{D}$, such that
robot $i$ is said to lie in the subdomain at time $t$ if
$p_{i}(t)\in\mathcal{S}(t)$. Let \(\partial \mathcal{S}(t)\) denote the boundary of the subdomain at time \(t\), and let \(q(t)\) be differentiable for almost every \(q\in\partial \mathcal{S}\).

We will use the locational cost \cite{locationalCost} as a metric of the
coverage performance in the subdomain $\mathcal{S}(t)$ at time $t$:
\begin{equation} \label{eqlocationalCost}\mathcal{H}(p(t),t) = \sum_{i=1}^{n}
\int_{V_{i}(p(t),t)} \left\|  p_{i}(t)-q\right\|  ^{2} \phi(q,t) \,dq
\end{equation} where $p(t)=\left[  p_{1}^{T}(t),\ldots,p_{n}^{T}(t)\right]
^{T}$ is the configuration of the multi-robot team and
$\phi:\mathcal{S}(t)\times[0,\infty)\to(0,\infty)$ is a density function that
captures the relative importance of the points in the subdomain at time $t$,
differentiable in both arguments. The subdomain is partitioned into regions of
dominance, and these form a proper partition of the subdomain. We utilize a
Voronoi tessellation of the domain, given by \begin{align*} V_{i}(p,t) =
\left\lbrace q\in\mathcal{S}(t)~\middle\vert~\left\| p_{i}-q\right\|
\leq\left\|  p_{j}-q\right\|  ~\forall j\right\rbrace \end{align*} where for
ease of notation we have dropped the explicit time dependency on the
configuration of the multi-robot system.

\subsection{Centroidal Voronoi Tessellations}

A necessary condition for the minimization of the locational cost in
\eqref{eqlocationalCost} is known to be that the agents form a centroidal
Voronoi tessellation (CVT) of the domain \cite{locationalCost}, i.e., \[
p_{i}(t) = c_{i}(p,t)\qquad\forall i \] where we define  $c_{i}(t)\in
V_{i}(p,t)$ to be the center of mass of Voronoi cell $i$ at time $t$, given by
\begin{align} c_{i}(p,t)  &  = \frac{\int_{V_{i}(p,t)}
	q\phi(q,t)\,dq}{m_{i}(p,t)} \label{eqcenterOfMass}%
\end{align} where $m_{i}(p,t)$ is the mass of the corresponding cell, \begin{align} m_{i}(p,t)  &  = \int_{V_{i}(p,t)} \phi(q,t)\,dq.
\label{eqmass} \end{align}

\subsection{Coverage Control Law} \label{CoverageControlLaw}

In \cite{leeTVDC}, a control law was proposed which was later shown in
\cite{SwarmBook} to achieve exponential converge to a CVT in the case of
time-varying densities. This control law was called TVD-C for
\emph{time-varying densities, centralized case}, given by
\begin{equation}
\label{eqTVD-C}\dot{p} = \left(  I - \frac{\partial c}{\partial p}\right)
^{-1}\left(  \kappa(c(p,t)-p) + \frac{\partial c}{\partial t}\right)
\end{equation} 
where the tuning parameter $\kappa>0$ controls the
exponential convergence rate, and 
\begin{align*}
	c(p,t) = \left[ c_{1}(p,t)^{T},\ldots,c_{n}(p,t)^{T}\right]  ^{T}.
\end{align*}
When many agents are used, the control law 
called TVD-D\textsubscript{1}, which stands for \emph{time varying densities,
	decentralized case with 1-hop adjacency information}, bypasses difficulties with computing the matrix inverse in \eqref{eqTVD-C} by approximating it with the truncated Neumann series,
\begin{equation}
\label{eqTVD-D1}\dot{p} = \left(  I + \frac{\partial c}{\partial p}\right)
\left(  \kappa(c(p,t)-p) + \frac{\partial c}{\partial t}\right).
\end{equation} 

The matrix $\frac{\partial c}{\partial p}$ in \eqref{eqTVD-C} and \eqref{eqTVD-D1} is a block
matrix that possess the sparsity structure of the Delaunay graph associated
with the Voronoi tessellation \cite{SwarmBook}. The $ij$\textsuperscript{th}
block is given by 
\begin{equation} \label{eqdcidpj}\left[  \frac{\partial
	c}{\partial p}\right]  _{ij} \!\!\!\!= \frac{\partial c_{i}}{\partial p_{j}} =-
\frac{\int_{\partial V_{ij}%
		(p,t)}(q-c_{i})(q-p_{j})^{T}\phi(q,t)\,dq}{m_{i}(p,t)\|p_{i}-p_{j}\|}%
\end{equation} when $i\neq j$, and 
\begin{equation} \label{eqdcidpi}\left[
\frac{\partial c}{\partial p}\right]  _{ii} \!\!\!\!= \frac{\partial c_{i}}{\partial
	p_{i}}= \sum_{j\in\mathcal{N}_{V_{i}}} \frac{\int_{\partial
		V_{ij}(p,t)}(q-c_{i})(q-p_{i})^{T}\phi(q,t)\,dq}%
{m_{i}(p,t)\|p_{i}-p_{j}\|}%
\end{equation}
where $\partial V_{ij}=V_{i}\cap V_{j}$, and \(\mathcal{N}_{V_i}\) is the Delaunay graph neighbor set of agent \(i\). These are
($d$$-$$1$)-dimensional integrals (e.g., line integrals if
$\mathcal{D} \subseteq\mathbb{R}^{2}$), and are zero if $\partial
V_{ij}=\emptyset$ (i.e., two cells are not adjacent), or if the shared boundary
has zero ($d$$-$$1$)-dimensional measure (e.g., points
in 2D).

Assuming that the domain of interest is static, the partial $\frac{\partial
	c}{\partial t}$ in \eqref{eqTVD-C} is given by $\frac{\partial c}{\partial t}
= \left[  \frac{\partial c_{1}}{\partial t}^{T},\ldots,\frac{\partial c_{n}%
}{\partial t}^{T}\right]  ^{T}$, where \begin{equation}
\label{eqdcdt1}\frac{\partial c_{i}}{\partial t} = \frac{\int_{V_{i}%
		(p,t)}(q-c_{i})\frac{\partial\phi}{\partial t}\,dq}{m_{i}(p,t)}.
\end{equation}

Although the control law in \eqref{eqTVD-C} was derived with consideration to
time-varying densities over static domains, the proof of exponential convergence in
\cite{SwarmBook} is still valid when the domain of interest is
time-varying with almost everywhere differentiable boundary, as long as
$p_{i}(t_{0})\in\mathcal{S}(t_{0})~\forall i$. However, additional terms are
needed in \eqref{eqdcdt1} to capture the evolution of the subdomain. In the
following sections, we will derive the needed terms such that \eqref{eqTVD-C}
is able to provide exponential convergence to a CVT configuration, even when
the subdomain of interest is time-varying.

\section{Time-Varying Domains}

\label{TimeVaryingDomains} In order to retain the exponential convergence to a
CVT property of control law \eqref{eqTVD-C}, we will need to add additional
terms in \eqref{eqdcdt1} to capture the evolution of the time-varying domain.
For the subsequent analysis, it is assumed that $\frac{dq}{dt}$ exists and is
known for almost every $q\in\partial\mathcal{S}(t)$.

Note that the center of mass integrals in \eqref{eqcenterOfMass} are time
dependent in both the integral kernel and the domain of integration. Thus,
Leibniz integral rule is needed.

\begin{lemma} [Leibniz Integral Rule\cite{du1999centroidal}]\label{lem:Leibniz} 
	Let $\Omega(p,t)$ be a region that
	depends smoothly on $t$ and that has a well-defined boundary
	$\partial\Omega(p,t)$. If $F(p,t)=\int_{\Omega (p,t)}f(p,t,q)\,dq$ for
	differentiable $f$, then \begin{align*} \frac{\partial F}{\partial t} =
	\int_{\Omega(p,t)}\frac{\partial f}{\partial t}\,dq +
	\int_{\partial\Omega(p,t)}f(p,t,q)\frac{\partial q}{\partial t}%
	^{T}\hat{n}(q)\,dq \end{align*} where $\hat{n}(q)$ is the unit outward normal for $ q\in\partial\Omega(p,t)$.
\end{lemma}

Application of Lemma \ref{lem:Leibniz} to \eqref{eqcenterOfMass} yields the
following result.

\begin{prop} [Time-Varying Domains]\label{prop:dcdt}
	
	Let $\mathcal{S}(t)$ be the convex subdomain to be covered with boundary
	$\partial\mathcal{S}(t)$. Assume that $\frac{d q}{d t}$ exists and is known
	for almost every $q\in\partial\mathcal{S}(t)$. Let $\partial\mathcal{S}%
	_{i}(p,t)=\partial\mathcal{S}(t)\cap V_{i}(p,t)$. Then the partial derivative
	$\frac{\partial c_{i}}{\partial t}$, which captures the time evolution of the
	center of mass of the $i$\textsuperscript{th} Voronoi cell due to the motion
	of the subdomain and the change in density function, is given by
	\begin{align}\label{eqdcidt2}
	\frac{\partial c_{i}}{\partial t} =
	\frac{1}{m_{i}(p,t)}\int_{V_{i}(p,t)} (q-c_{i})\frac{\partial\phi}{\partial
		t}\, dq + \frac{1}{m_{i}(p,t)} \int_{\partial\mathcal{S}_{i}(p,t)}
	(q-c_{i})\phi(q,t) \frac{d q}{d t}^T\hat{n}\,dq. 
	\end{align}
\end{prop}
\begin{proof} By the product rule, we find 
	\begin{align*} \frac{\partial
		c_i}{\partial t} = \frac{\frac{\partial}{\partial
			t}\left(\int_{V_i(p,t)}q\phi(q,t)\,dq\right)}{m_i(p,t)}
	-\frac{c_i(p,t)}{m_i(p,t)}\frac{\partial m_i}{\partial t}. 
	\end{align*} 
	Both terms now contain derivatives of integrals with time-varying integral
	kernels and domains. Applying Leibniz integral rule to the first term yields
	\begin{align*} \frac{\partial}{\partial
		t}\left(\int_{V_i(p,t)}q\phi(q,t)\,dq\right)= \int_{V_i(p,t)} q\frac{\partial
		\phi}{\partial t}\, dq  + \int_{\partial V_i(p,t)} q\phi(q,t) \frac{\partial
		q}{\partial t}^T\hat{n}\,dq 
	\end{align*} 
	Then the latter integral can be split
	into two terms: for \(q\in\partial\mathcal{S}_i(p,t)\), and for \(q\in\partial
	V_i(p,t)\backslash\partial\mathcal{S}_i(p,t)\). Note that only the points on
	the boundary of the subdomain change explicitly with time, i.e., since
	\(\frac{\partial q}{\partial t}=0\) for almost every point \(q\notin\partial
	\mathcal{S}_i(p,t)\) we get 
	\begin{align*} \frac{\partial}{\partial
		t}\left(\int_{V_i(p,t)}q\phi(q,t)\,dq\right)= \int_{V_i(p,t)} q\frac{\partial
		\phi}{\partial t}\, dq + \int_{\partial\mathcal{S}_i(p,t)} q\phi(q,t) \frac{d
		q}{d t}^T\hat{n}\,dq. 
	\end{align*} 
	Similarly, for \(\frac{\partial
		m_i}{\partial t}\) we find 
	\begin{align*} \frac{\partial m_i}{\partial t}=
	\int_{V_i(p,t)} \frac{\partial \phi}{\partial t}\, dq +
	\int_{\partial\mathcal{S}_i(p,t)} \phi(q,t) \frac{d q}{d t}^T\hat{n}\,dq.
	\end{align*} 
	Collecting like terms, we get
	\begin{align*} 
	\frac{\partial c_i}{\partial t} = 
	\frac{\textstyle\int_{V_i(p,t)} (q-c_i)\frac{\partial \phi}{\partial t}\,
		dq}{m_i(p,t)} + \frac{\textstyle
		\int_{\partial\mathcal{S}_i(p,t)} (q-c_i)\phi(q,t) \frac{d q}{d t}^T\hat{n}\,dq
	}{m_i(p,t)}
	\end{align*}  
	as was to be shown.\hfill
\end{proof}

\begin{cor} [Exponential Convergence]\label{cor:ExponentialConvergence}
	The control law in \eqref{eqTVD-C} with updated partial derivative as in
	Proposition \ref{prop:dcdt} yields exponential convergence to a CVT with
	exponential decay rate controlled by $\kappa>0$ over smoothly time-varying
	domains and densities, as long as \(p_i(t_0)\in\mathcal{S}(t_0)~\forall i\).
\end{cor}

\begin{proof} Same as in \cite{SwarmBook}.\hfill
\end{proof} Proposition
\ref{prop:dcdt} introduces a new term in the computation of $\frac{\partial
	c_{i}}{\partial t}$ when compared to \eqref{eqdcdt1}, which allows us to
explicitly take into account the evolution of the subdomain. This additional
term can serve as an exogenous input to the multi-robot team to control their
collective position and scale, while the density function can be chosen to
provide the desired shape in multi-robot team formations, as suggested in
\cite{D3C_HSI}. These choices of inputs possess the
advantages of being flexible with respect to the size of the system, and as well as being
identity-agnostic, so that the input may be chosen without the need of
performing assignments of roles, such as is needed in leader-follower schemes.

In the next section, we explore the coverage problem over uniform densities in
convex 2-polytope subdomains, and analytical expressions are found for the
terms in \eqref{eqTVD-C} and \eqref{eqTVD-D1}.

\section{Coverage over Convex 2D-Polytopes with Uniform Density}

\label{AnalyticalFormulas} We now focus on the class of coverage problems where
the subdomain is a convex polytope in $\mathbb{R}^{2}$ with \(N\)
vertices using uniform density, i.e., with $\phi(q,t)=1$ for all
$q\in\mathcal{S}(t)$ and $t\geq 0$. As an immediate consequence of this, the first
term in expression provided in Proposition \ref{prop:dcdt} becomes zero. For
this class of problems, we will further see that we can find analytical
expressions in terms of time, neighboring agent positions and boundary conditions in that agent's Voronoi cell. The control law can be computed solely on local neighbor information and broadcasted domain information. We begin by considering the vertices of a Voronoi
cell.

\subsection{Vertices of the Voronoi Cell} Assuming that agent $i$ has $h'$
neighbors, the $i$\textsuperscript{th} Voronoi cell will be a convex
2-polytope with at least $h'$$-$$1$ vertices. To account for vertices due to the
intersection of the Voronoi cell with the the subdomain boundary, we will
assume the Voronoi cell consists of $M\geq h'$$-$$1$ vertices. We will denote the
vertex due to neighbors \(j\) and \(k\) (which can be determined by employing a Delaunay triangulation \cite{lee1980two}) as \(V_{ijk}\). The vertex location
\(V_{ijk}\) is given by the circumcenter of the triangle formed by connecting
the position of these three agents \(p_i\), \(p_j\) and \(p_k\), and can be
found analytically in terms of these
\cite{pedoe1995circles,cortes2004coverage}. We present the equation below for
the sake of completion,
\begin{align*}
V_{ijk}=\frac{1}{2}\frac{\alpha_{i}p_{i}+\alpha_{j}p_{j}+\alpha_{k}p_{k}%
}{||{p_{ij}}||^{2}||{p_{jk}}||^{2}-(p_{ij}^{T}p_{jk})^{2}}
\end{align*}
where the agents are assumed to be oriented in a
counterclockwise order, where \(p_{ab}= p_b-p_a\), and where
\begin{align*}
\alpha_{i}=\|{p_{jk}}\|^{2}p_{ij}^{T}p_{ik},
\alpha_{j}=\|{p_{ik}}\|^{2}p_{ij}^{T}p_{kj},
\alpha_{k}=\|{p_{ij}}\|^{2}p_{ik}^{T}p_{jk}.
\end{align*}

Denote \(V_{i\mathcal{S}\ell}\) to be the \(\ell\)\textsuperscript{th} vertex
of the \(i\)\textsuperscript{th} Voronoi cell that is on the subdomain
boundary, where without loss of generality these are assumed to be ordered
counterclockwise. These may be domain polytope vertices, or may be vertices
due to the intersection of the Voronoi cell with the subdomain. As the former vertices are assumed to be known, we provide an analytical expression for the latter.

Any point \(q\) on a face of the subdomain may be expressed using the
vertices that define the face via the parameterization \(q = L_{\ell}(\tau) =
\partial\mathcal{S}^{\ell} +
(\partial\mathcal{S}^{(\ell+1)}-\partial\mathcal{S}^{\ell})\tau\) for certain
\(\tau\in[0,1]\), where without loss of generality the vertices of the
subdomain \(\partial \mathcal{S}^\ell\) are assumed to be ordered
counterclockwise, and where \(\partial \mathcal{S}^{N+1} =
\partial\mathcal{S}^1\). Let any point \(q\in \partial V_{ij}\), the interior
Voronoi face due to two agents \(i\) and \(j\) that intersects polytope edge
\(L_\ell\), be given by \(q = R_{ij}(s) = \frac{1}{2}(p_i+p_j) + t_{ij} s\) for
certain \(s\in\mathbb{R}\), where \(t_{ij}=S(p_j-p_i)\) with \(S =
\left[\begin{smallmatrix} 0&-1\\1&0 \end{smallmatrix}\right]\) being a skew
symmetric rotation matrix. The sough after vertex \(V_{i\mathcal{S}j}\) is given by
\begin{align*}
V_{i\mathcal{S}j} = \partial\mathcal{S}^{\ell} + \left(\partial\mathcal{S}^{(\ell+1)}-\partial\mathcal{S}^{\ell}\right)\tau^*
\end{align*}
which exists if 
\begin{align*}
\tau^* = \frac{(\frac{1}{2}(p_i+p_j)-\partial\mathcal{S}^{\ell})}{\left(\partial\mathcal{S}^{(\ell+1)}-\partial\mathcal{S}^{\ell}\right)^T(p_j-p_i)}\in[0,1]
\end{align*}
found from solving for \(\tau\) in \(L_\ell(\tau)=R_{ij}(s)\).

The mass and center of mass can now be computed analytically 
in terms of the vertices of the Voronoi cell.

\subsection{Mass and Center of Mass of a Voronoi Cell}

\label{CentroidofVC} For convenience, denote the vertices of the \(i\)\textsuperscript{th} Voronoi cell as \(V_{ij}(p,t)\), \(j\in\{1,\ldots,M\}\), which are assumed without loss of generality to be ordered counterclockwise. Then the mass and center of mass may be computed as \cite{cortes2004coverage}
\begin{align*}
m_i(p,t) = \frac{1}{2}\sum_{j=1}^{M-1} V_{i(j+1)}^TSV_{ij}
\end{align*}
\begin{equation*}
c_i(p,t) = \frac{1}{6m_i(p,t)}\sum_{j=1}^{N-1}\left(V_{i(j+1)}+V_{ij}\right)
\left(V_{i(j+1)}^TSV_{ij}\right).
\end{equation*} 


\subsection{Partial Derivatives of a Centroid}\label{AnalyticalPartialDerivatives}

For a convex 2-polytope subdomain, every Voronoi cell is also a convex 2-polytope, and the boundary integrals in \eqref{eqdcidpj}, \eqref{eqdcidpi}, and \eqref{eqdcidt2} may be solved along each face analytically by using the parametric line integral as in the following lemma.
\begin{lemma} [Parametric Line Integral]\label{arclength} Let the line integral
	of $f(q)$ along $L$ be denoted by $\int_{L} f(q)dq$ where the point \(q=[q_1,q_2]^T\) can be parameterized such that if \(q_1 = h(\tau)\) and \(q_2 = g(\tau)\) for some \(\tau\in[t_0,t_1]\), then
	\begin{align*} 
	\int_{L} f(q)dq =
	\int_{t_{0}}^{t_{1}} f([h(\tau),g(\tau)]^T)||{q}^{{\prime}}(\tau)||dt
	\end{align*} 
	where
	$||{q}^{{\prime}}(\tau)||=\sqrt{(\frac{dh}{d\tau})^{2}+(\frac{dg}{d\tau})^{2}}$. 
\end{lemma}

Let that \(\partial\mathcal{S}_i(p,t)=\partial\mathcal{S}(t)\cap V_i(p,t)=\bigcup_{j}\partial\mathcal{S}_i^j\), where \(\partial\mathcal{S}_i^j\) is the \(j\)\textsuperscript{th} Voronoi cell face shared with the subdomain boundary. Each of the \(\partial\mathcal{S}_i^j\) is a straight line, connected by two Voronoi cell vertices \(v_{1}^j\) and \(v_{2}^j\). By linearity of the integral, the partial derivative \eqref{eqdcidt2} may be computed as
\begin{align*} 
\frac{\partial c_{i}}{\partial t} =
\sum_{j}\frac{\nu_{j}(t)}{m_{i}}%
\int_{\partial S_{i}^{j}} (q-c_{i}) \, dq 
\end{align*} 
where \(\nu_{j}(t)\) is the velocity of the points in the boundary $\partial S_{i}^{j}$, assumed to be the same for every point, projected in the unit outward direction (i.e., \(\nu_{j}(t)=\pm\frac{dq}{dt}^TS\frac{v_{2}^j-v_{1}^j}{\|v_{2}^j-v_{1}^j\|} \forall q\in \partial S_{i}^{j}\) depending on the orientation of the vertices).
\begin{lemma}[Analytic Partial Derivatives]\mbox{}

	Under uniform density and a convex 2-polytope subdomain with \(\partial\mathcal{S}_i^j=\overline{v_{1}^jv_{2}^j}\), \eqref{eqdcidt2} can be expressed as
	\begin{align*}
	\frac{\partial c_{i}}{\partial t} = \sum_j\frac{\nu_{j}(t)}{2m_{i}%
	}\left(  \left(  v_{2}^{j}-c_{i}\right)  +\left(  v_{1}^{j}-c_{i}\right)
	\right) \|v_{2}^{j}-v_{1}^{j}\|
	\end{align*}
	and for \(\partial V_{ij}=\overline{v_{1}^jv_{2}^j}\), \eqref{eqdcidpj} may be expressed as
	\begin{multline*}
	\frac{\partial c_{i}}{\partial p_{j}} =
	-\frac{\|v_{2}^{j}-v_{1}^{j}\|}%
	{m_{i}\|p_{i}-p_{j}\|}\Big[\left(  v_{1}^{j} - c_{i}\right)  \left( v_{1}^{j}
	- p_{j}\right)  ^{T} +\frac{1}{2}\left(  (v_{1}^{j} - c_{i})(v_{2}^{j} -
	v_{1}^{j})^{T} + (v_{2}^{j} -v_{1}^{j})(v_{1}^{j} - p_{j})^{T} \right) \\
	+\frac{1}{3}\left(  v_{2}^{j} -v_{1}^{j} \right)  \left(  v_{2}^{j} - v_{1}^{j}
	\right)^{T}\Big].
	\end{multline*}
\end{lemma}
\begin{proof}
	Assume the following parameterization of the
	points $q\in\partial S_{i}^{j}=\overline{v_{1}^jv_{2}^j}$ ,
	\begin{align*} 
	q(\tau) = v_{1}^{j} (1-\tau) +
	v_{2}^{j} \tau,
	\:\:\tau\in[0,1]. 
	\end{align*} 
	It follows that
	$q^{\prime}= (v_{2}^{j}-v_{1}^{j})$ for all $\tau$, and consequently $dq =
	\|v_{2}^{j}-v_{1}^{j}\|d\tau$. Applying Lemma \ref{arclength} to the integral in \eqref{eqdcidt2}, 
	\begin{align*}
	\frac{\partial c_{i}}{\partial t}  
	&  =
	\sum_j\frac{\nu_{j}(t)}{m_{i}}\int_{0}^{1} \left(  v_{1}^{j} (1-\tau) +
	v_{2}^{j} \tau- c_{i}\right)  \|v_{2}^{j}-v_{1}^{j}\|\,d\tau.
	\end{align*}
	The result follows from expanding and integrating the terms. A similar process can be applied to \eqref{eqdcidpj} by instead using the vertices that make up the boundary for \(\partial V_{ij}\) in the parameterization.\hfill
\end{proof}

\section{Scalability of the Algorithm} \label{Scalability}
In this section, we will argue that the proposed algorithm is scalable for a large group of robots (e.g., robot swarms). As a well-known fact in the field of computational geometry, the number of edges meeting at a vertex in the Voronoi diagram is not less than three, and each edge connects two vertices; it follows Euler's relation that the average number of edges of a Voronoi cell equals six\cite{aurenhammer1991voronoi,ohya1984improvements}. This fact is the key for our algorithm's significantly reduced computational cost.
\subsection{Complexity Analysis of the Algorithm}
Here we analyze the time complexity of the proposed algorithm for a single agent $p_i$.
\begin{enumerate}
	\item There are many well-studied algorithms to construct Voronoi diagrams, such as the Incremental Insertion Method and the Divide and Conquer Method\cite{aurenhammer1991voronoi}. The quaternary incremental algorithm builds the Voronoi diagram for $n$ seeds with average complexity $O(n)$ despite its worst-case complexity $O(n^2)$, and the Divide and Conquer Method has an average and worst-case complexity $O(n\log n)$\cite{ohya1984improvements}. In our algorithm, agents only require the density and boundary information within their own
	Voronoi cells (which can be computed in a distributed way
	\cite{cortes2004coverage}) and their Delaunay neighbor positions to compute the
	quantities in \eqref{eqTVD-C}. As described in \cite{leeTVDC}, the matrix inverse
	in TVD-C can be approximated using the truncated Neumann series expansion so that an
	agent only use as much information about the multi-robot team as it has access
	to. For one agent in our algorithm, only a partial Voronoi diagram, which is generated by the agent itself, its neighbors ($|\mathcal{N}_i|=6$ on average), and its neighbors' neighbors (to enclose the Voronoi cells of $\mathcal{N}_i$, a total of 12 on average) is needed instead of the overall Voronoi diagram. Thus, the time complexity becomes $O\left(\left|\bigcup_{j\in\mathcal{N}_i}\mathcal{N}_j\cup \mathcal{N}_i\right|\right)$, which is on average constant time, by using the quaternary incremental algorithm.
	\item We can obtain necessary information for $\frac{\partial c}{\partial p}$ in the control law (i.e., masses $m_i$ and $m_j$, centroids $c_i$ and $c_j$ of Voronoi cell $V_i$ and $V_j \in \mathcal{N}_i$) in constant time thanks to analytic expressions derived in previous section. The complexity becomes $O(1)$ on average in light of the fact of average number of neighbors is independent of the size of the robot team.
	\item The complexity of computing $\frac{\partial c_i}{\partial t}$ is determined by the number of edges of its Voronoi cell (\(O(M)\)). Similarly to the previous cases, as the number of edges depends on the number of neighbors, the complexity is constant time on average.
	\item Given the above terms, the control law for agent $p_{i}$ can be calculated with time complexity $O(1)$.
\end{enumerate}

\begin{table}[tb]\centering\renewcommand{\arraystretch}{1.25}
	\caption{Number of agents (out of 100) categorized by their neighbor sets $\mathcal{N}_i$ (i.e., number of neighbors).} \label{tab:NoofNeighbors}
	\begin{tabular}{c|c|c|c|c|c|c|}
		\hline
		\multicolumn{1}{|c|}{\textbf{Size of $\mathcal{N}_i$}} & \textbf{2} & \textbf{3} & \textbf{4} & \textbf{5} & \textbf{6} & \textbf{7} \\ \hhline{=======}
		\multicolumn{1}{|c|}{\textbf{Simulation}}                           & \multicolumn{6}{c|}{\textbf{Number of Agents}}                                 \\ \hline
		\multicolumn{1}{|c|}{\textbf{Trial 1}}          & 4          & 1          & 31         & 10         & 47         & 7          \\ \hline
		\multicolumn{1}{|c|}{\textbf{Trial 2}}          & 1          & 4          & 29         & 9          & 50         & 7          \\ \hline
		\multicolumn{1}{|c|}{\textbf{Trial 3}}          & 2          & 3          & 28         & 9          & 52         & 6          \\ \hline
		\multicolumn{1}{|c|}{\textbf{Trial 4}}          & 3          & 4          & 28         & 8          & 51         & 6          \\ \hline
		\multicolumn{1}{|c|}{\textbf{Trial 5}}          & 3          & 3          & 27         & 9          & 52         & 6          \\ \hhline{=======}
		\multicolumn{1}{|c|}{\textbf{Average}}          & 2.6        & 3          & 28.6       & 9          & 50.4       & 6.4        \\ \hline
	\end{tabular}
\end{table}
\subsection{Validation of Scalability}
We show the asymptotic nature of the average number of neighbors and the efficiency of the control law for a large group of robots via numerical simulations. The simulation of the proposed control strategy was run for five trials using 100 robots with random initial configurations, and the outcomes are collected in Table \ref{tab:NoofNeighbors}. Around fifty percent of agents have 6 neighbors in the simulations. The majority of the remaining agents actually had less, with only about 7$\%$ having no more than seven neighbors.This percentage increased with the size of the swarm.

Moreover, simulations of the proposed control law with and without term $\frac{\partial c}{\partial t}$ are conducted. The simulation is shown in Fig. \ref{Simulation}i, where a team of 100 robots is performing coverage of a subdomain which is tracking a circular trajectory. The instantaneous aggregated CVT error $e_{a}(t)=\|p(t)-c(p,t)\|$ of the proposed control strategy with and without the term $\frac{\partial c}{\partial t}$ is shown in Fig. \ref{Simulation}ii. As we can tell from the results, the steady state mean values of the simulation with and without $\frac{\partial c}{\partial t}$ are 0.0562 m and 0.1248 m respectively. The aggregated CVT error 0.0562 m for 100 robots (i.e., 0.000562 m on average for each robot) suggests that the proposed control law converges to a CVT. The $\frac{\partial c}{\partial t}$ term improves the error performance by 54.97$\%$ in this scenario with a large group of agents. The error is not exactly zero because the distributed form of the control law \eqref{eqTVD-D1} was used, which does not perform feedforward on every agent in the team. Thus, the agents on the boundary can immediately follow the motion of subdomain as they have direct access to the boundary conditions (i.e., whose Voronoi cell boundaries have overlapped edges with subdomain's boundaries), but the agents on the interior must wait for the boundary conditions to propagate through feedback. 

\begin{figure}[ptb] \centering 
\begin{tabular}{>{\centering}m{0.4\linewidth} m{0.6\linewidth}<{\centering}}
	\includegraphics[width=\linewidth]{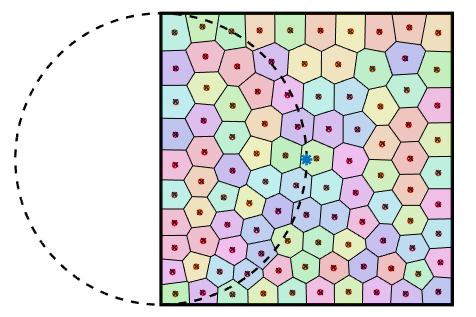} & \includegraphics[width=\linewidth]{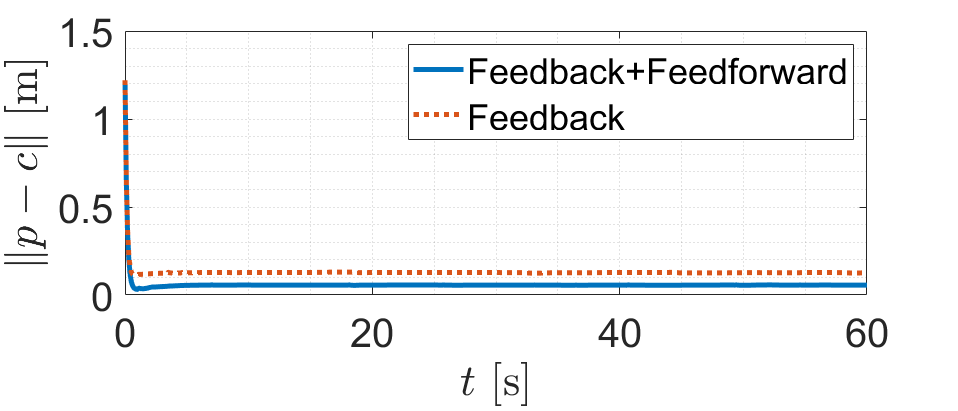}\\
	(i) & (ii)
\end{tabular}	
	\caption{Simulation of a group of 100 single-integrator robots (i); the center of subdomain (asterisk) moves along the circular trajectory (dash line) in counterclockwise fashion. Results of the simulation of the control law with and without the additional feedforward term $\frac{\partial c}{\partial t}$ (ii).\label{Simulation}}
\end{figure}

\section{Multi-Robot Implementation}
\label{Implementation}

\begin{figure*}[h] \centering
	\begin{tabular}{>{\centering}m{0.45\linewidth} m{0.45\linewidth}<{\centering}}
		\includegraphics[width=\linewidth]{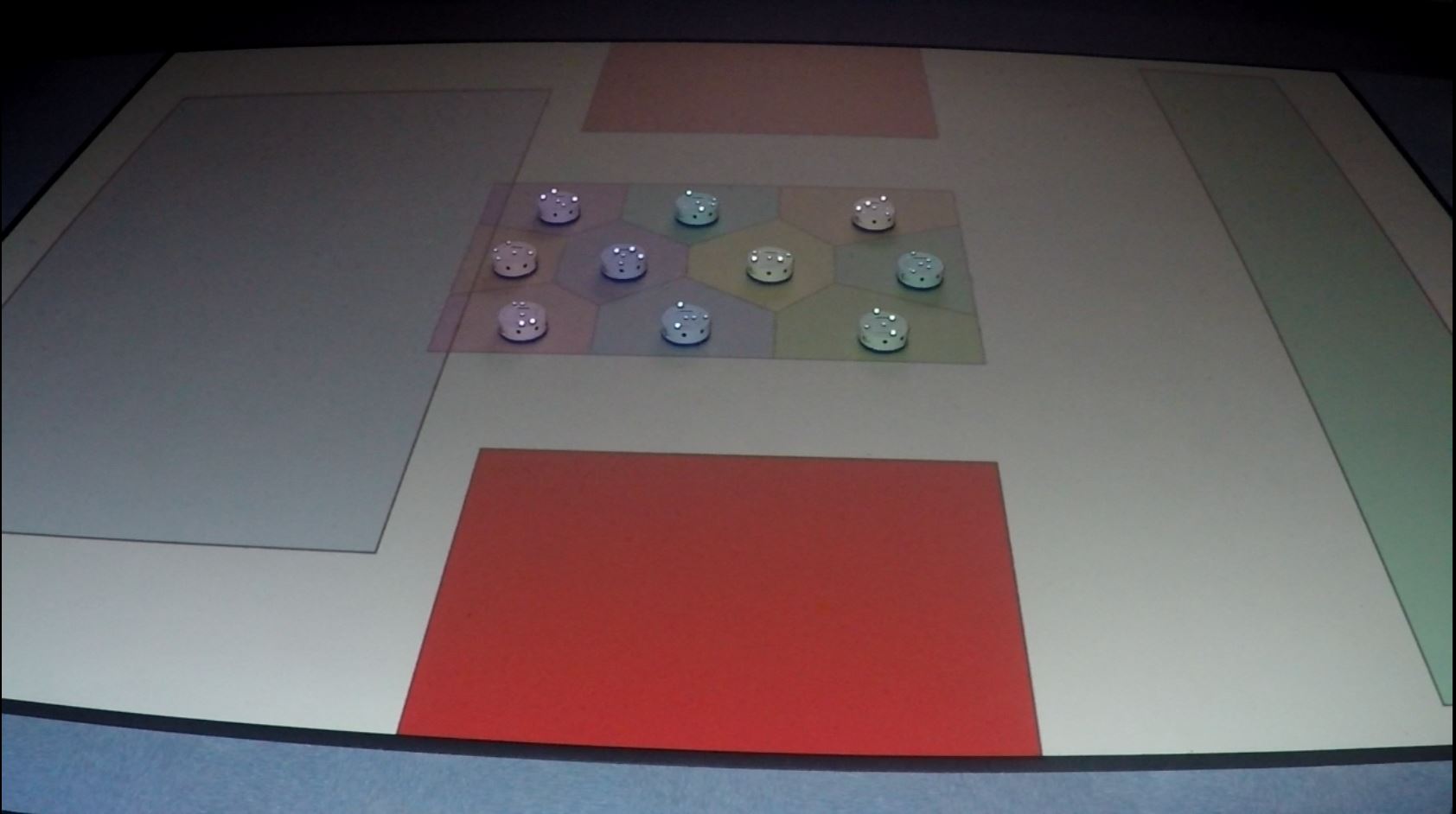} & \includegraphics[width=\linewidth]{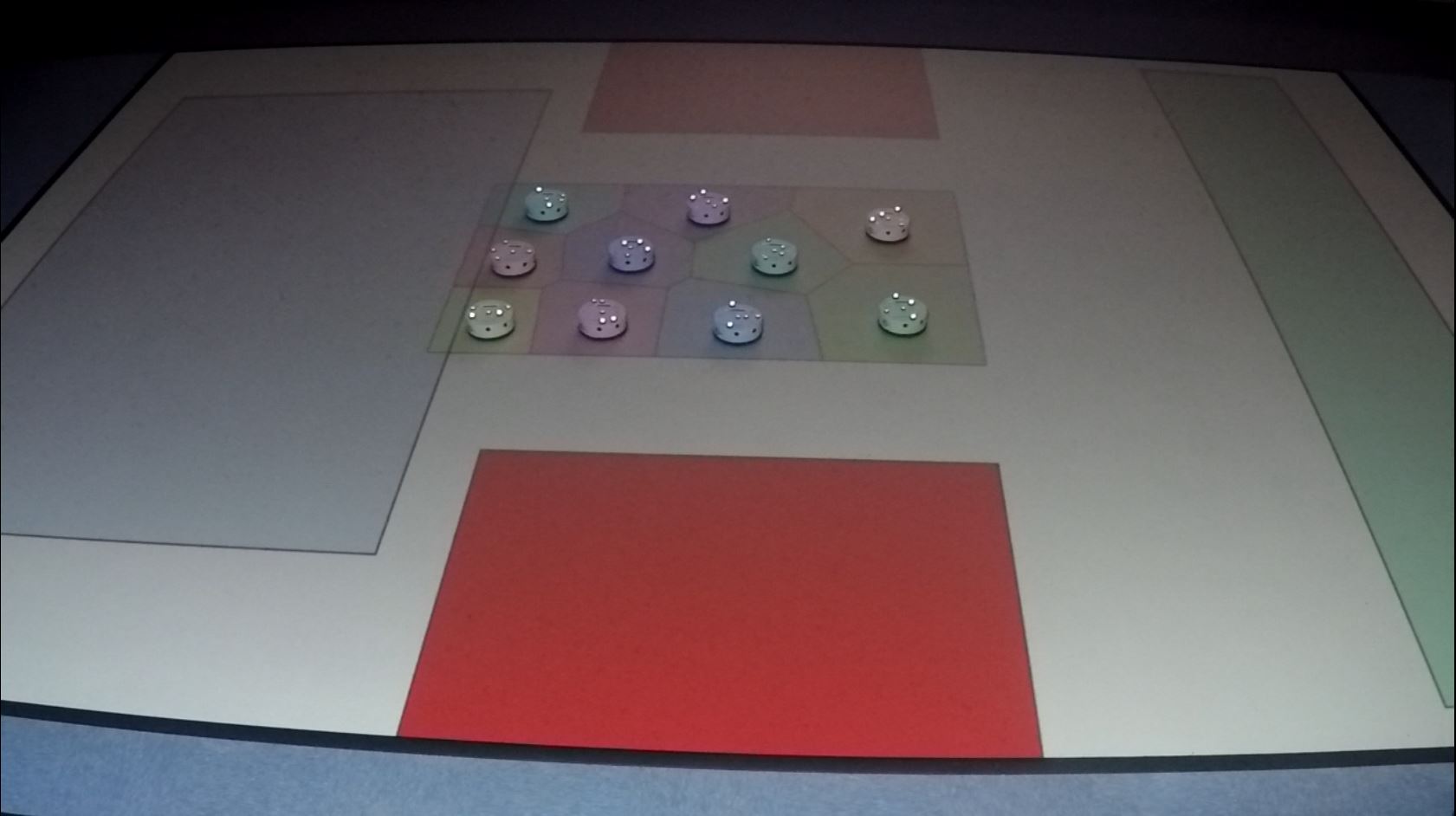}\\
		(i) & (ii)
	\end{tabular}
	\caption{Multi-robot implementation of control law with (i) and without (ii) term $\frac{\partial c}{\partial t}$. 10 robots travel from left end of workspace (blue box) towards the right end (green box) while avoiding the obstacles (red boxes). The subdomain translates and scales over time. An overhead projector is used to visualize the subdomain, goal location and obstacles in real time. \label{fig:Robots}}
\end{figure*}

\begin{figure}[ptb] \centering
	\includegraphics[width=0.6\linewidth,clip=true,trim=0 5 45 30]{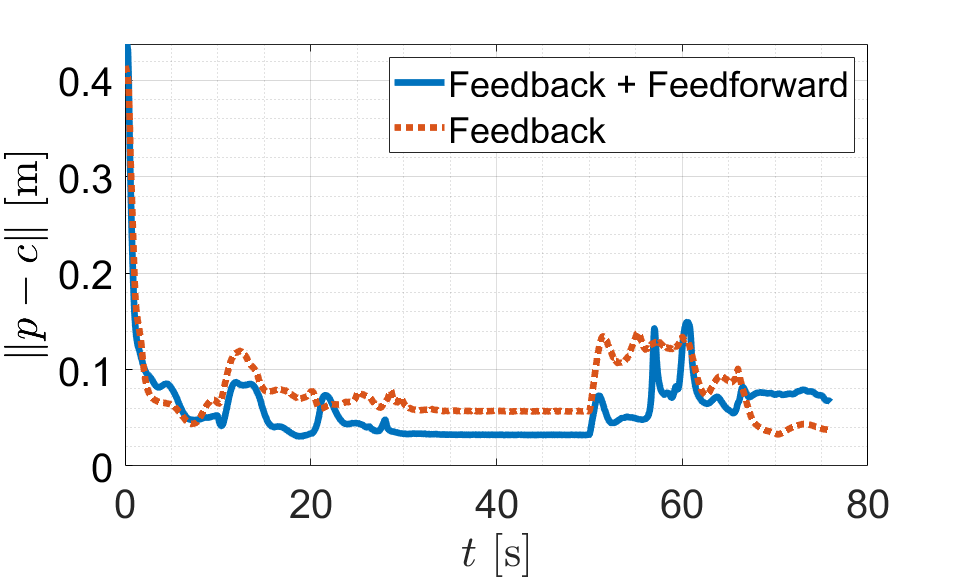}
	\caption{Results of experiments. Profiles of aggregated error from optimal configuration with and without term $\frac{\partial c}{\partial t}$  are indicated by blue solid line and red dot line}\label{ResultsofExperiment}
\end{figure}
In section \ref{AnalyticalFormulas}, the coverage control laws \eqref{eqTVD-C} and \eqref{eqTVD-D1} are modified to account for time-variations in the subdomain and the terms needed for the
equations are provided analytically in the case of uniform coverage on a convex 2-polytope domain. The results were simulated on a large team of robots in Section \ref{Scalability}.
In this section, we implement the proposed control strategy on a group of real
mobile robots. The experiments are carried out on a workstation with
\emph{Intel(R) Xeon(R) W-2125 processor}, and the algorithms are implemented in
\emph{MATLAB}. The multi-robot team consists of 10 \emph{Khepera IV}
differential-drive mobile robots. A motion capture system consisting of 8 \emph{Vicon Vantage V8} cameras are used to provide real-time position information of the multi-robot team, and an overhead projector is used to visualize a virtual environment and Voronoi tessellation. The robot workspace is defined to be a 5.182 m \(\times\) 3.658 m (17 feet $\times$ 12 feet) rectangular area.

\subsection{Experiment Description}
An experiment is carried out to validate the proposed control
strategy with a rectangular subdomain \(\partial \mathcal{S}(t)= [x_{min}(t),x_{max}(t)]\times[y_{min}(t),y_{max}(t)]\) with uniform density. 
In the experiment, we synthesize control laws for the different agents by driving the coverage domain for a group of robots to achieve a motion plan. The experiment is shown in Fig. \ref{fig:Robots}. To achieve the task, the subdomain simultaneously translates and scales to go through the narrower passage and avoid the obstacles. The velocities for the subdomain boundaries are defined as piece-wise constants to achieve the desired scaling and translation. The metric used to determine performance of the proposed control strategy on the real robots is same as that in simulation.
The control law \eqref{eqTVD-D1} and analytic expressions presented in Section \ref{AnalyticalFormulas} are used for the experiment with a control gain \(\kappa=2\).

As mentioned in Corollary \ref{cor:ExponentialConvergence}, the aggregated
error $e_{a}$ is expected to decay exponentially to zero under control law TVD-C (almost exponential convergence under control law TVD-D\textsubscript{1}), up to the error introduced by the mapping of the control law to the differential-drive motion of the robots, wheel saturations, and delays. For comparison, the experiment is executed with and without the term $\frac{\partial c}{\partial t}$, to evaluate the contribution by the inclusion of the term.

\subsection{Experimental Results} 

From Fig. \ref{fig:Robots}, we can tell the difference between performances with and without $\frac{\partial c}{\partial t}$. The robots form and maintain an almost symmetric configuration with $\frac{\partial c}{\partial t}$ (Fig. \ref{fig:Robots}i) while the robots lag behind the motion of subdomain (Fig. \ref{fig:Robots}ii) and could fail to catch up if the subdomain moves faster. Moreover, the profiles of $\|p(t)-c(p,t)\|$ are plotted in Fig. \ref{ResultsofExperiment} for the experiments of implementing the control law with the inclusion of the $\frac{\partial c}{\partial t}$ term and the absence of it, and they are shown as blue solid line and red dot line respectively. The aggregated error $e_{a}$ greatly decreases at first, as expected due to the exponential convergence. As the agents approach their CVT configuration, delays, saturations, and errors in mapping velocities to differential-drive introduce small disturbances that result in variations along a small constant value. The mean of the steady-state values of the metric are 0.0323 m (i.e., 0.0032 m per robot --- well within the footprint of the robots which are 0.14 $m$ in diameter) and 0.057 m for with and without the inclusion of \(\frac{\partial c}{\partial t}\). There is a noticeable improvement in \(e_a\) of nearly 43.3\% due to the presence of the feedforward term.

\section{Conclusion}\label{Conclusion} 
An innovative control scheme is developed of synthesizing multi-agent controllers via coverage control on time-varying domains. The approach offers a capability of controlling the behavior of multi-agent systems by simply manipulating the domain to be covered, and it has the advantage of being agnostic to the size of the system. Analytic solutions for the uniform density case is obtained, and scalability of the approach is demonstrated. The feasibility of proposed control strategy is validated in simulation for a large number of agents and experimentally on a team of real differential-drive wheeled robots.

\bibliographystyle{unsrt} 
\bibliography{ACC_arXive}
\end{document}